\documentclass[a4paper,11pt]{article}
\usepackage{latexsym}
\usepackage{amsmath, amssymb,amsthm}
\usepackage{amsfonts,euscript}
\usepackage{setspace}
\usepackage{mathrsfs}

\newtheorem{theorem}{Theorem}

\newtheorem{lemma}[theorem]{Lemma}

\newtheorem{corollary}[theorem]{Corollary}
\newtheorem{remark}[theorem]{Remark}

\DeclareMathAlphabet{\mathpzc}{OT1}{pzc}{m}{it}

\def\R{\mathbb R}
\def\C{\mathbb C}
\def\S{\mathbb S}
\def\Dom{\mathcal D}
\def\Im{{\rm Im}}
\def\Re{{\rm Re}}
\def\e{{\rm e}}
\def\i{{\rm i}}
\def\Ln{{\rm Ln}}
\def\Q{L} 
\def\B{\mathcal B}
\def\Ju{J}

\renewcommand\bar{\overline}
\newcommand\cF{{\mathscr F}}

\def\bbbone{{\mathchoice {\rm 1\mskip-4mu l} {\rm 1\mskip-4mu l}
{\rm 1\mskip-4.5mu l} {\rm 1\mskip-5mu l}}}
\def\one{\bbbone}

\def\coker{{\rm cocker}}
\def\Ja{\mathcal J}
\def\Gammaq{\Gamma^{\square}}
\def\Wind{{\mathrm{Wind}}}
\def\ie{\emph{i.e.}}
\def\EE{\mathcal E}
\def\K{\mathcal K}
\def\t{{\scriptscriptstyle\#}}

\begin{document}

\title{Does Levinson's theorem count complex eigenvalues~?}
\author{F. Nicoleau${}^1$, D. Parra${}^2$, S. Richard${}^3$\footnote{On leave of absence from
Univ.~Lyon, Universit\'e Claude Bernard Lyon 1, CNRS UMR 5208, Institut Camille Jordan,
43 blvd.~du 11 novembre 1918, F-69622 Villeurbanne cedex, France.} \footnote{
Supported by JSPS Grant-in-Aid for Young Scientists A
no 26707005.}}

\maketitle

\vspace{-1cm}

\begin{quote}
\emph{
\begin{itemize}
\item[$^1$]
Laboratoire de Math\'ematiques Jean Leray,  UMR CNRS 6629,
Universit\'e de Nantes, 44322 Nantes Cedex 3, France
\item[$^2$]
Univ.~Lyon, Universit\'e Claude Bernard Lyon 1,
CNRS UMR 5208, Institut Camille Jordan,
43 blvd. du 11 novembre 1918, F-69622 Villeurbanne cedex, France
\item[$^3$] Graduate school of mathematics, Nagoya University,
Chikusa-ku,  Nagoya 464-8602, Japan
\item[] \emph{E-mail:} francois.nicoleau@univ-nantes.fr, parra@math.univ-lyon1.fr, richard@math.nagoya-u.ac.jp
\end{itemize}
}
\end{quote}

\vspace{-5mm}

\begin{abstract}
Yes it does~! Indeed an extended version of Levinson's theorem is proposed for a system
involving complex eigenvalues. The perturbed system corresponds to a realization
of the Schr\"odinger operator with inverse square potential on the half-line, while the Dirichlet Laplacian on the half-line is chosen for the reference system.
The resulting relation is an equality between the number of eigenvalues of the perturbed system and the winding number
of the scattering system together with additional operators living at $0$-energy and at infinite energy.
\end{abstract}


\section{Introduction}

Since its discovery by N. Levinson in 1949, the so-called \emph{Levinson's theorem}
has attracted a lot of interest and many researchers have generalized this relation between spectral and scattering theory.
In the simplest situations, it corresponds to an equality between the number of bound states of a quantum mechanical system
and an expression based on the scattering part of the system. Very often, this latter expression involves a regularization procedure,
and in many cases some corrections must also be taken into account.
Since the literature on the subject is very vast, we simply refer to the review papers \cite{Ma,R16} and to the references mentioned therein.

Up to our knowledge, all these investigations have taken place in the context of self-adjoint operators in a Hilbert space
since this framework is the natural one for quantum mechanics.
On the other hand, non self-adjoint operators have recently been deeply studied, and in this context
an extension of Levinson's theorem to non real eigenvalues seems a natural question.
The purpose of this note is precisely to exhibit such a relation for a system involving non self-adjoint operators with a finite family of complex eigenvalues.

The system we consider consists in the operator
$$
H=-\partial_x^2+\big(m^2-\frac14\big)\frac{1}{x^{2}}
$$
in $L^2(\R_+)$ with a complex parameter $m$ and with a possibly complex boundary condition at $x=0$. A large family of such operators
have been recently analyzed in \cite{DR} where their spectral and scattering theory have been constructed.
For the reference system $H_0$ we simply consider the Dirichlet Laplacian on $\R_+$, but other choice are possible
by using the chain rule. By particularizing some expressions obtained in \cite{DR}, explicit formulas for the wave operators
and for the scattering operators for the pair $(H,H_0)$ are directly available. Note that if $H$ is not self-adjoint,
the wave operators are not partial isometries and the scattering operator is not a unitary operator.
Nevertheless, these operators can be defined and their properties studied.

Our approach for proving a Levinson's type theorem is based on the topological
approach firstly introduced \cite{KR06,KR08} and extensively reviewed in \cite{R16}.
In this framework, Levinson's theorem corresponds to an index theorem in scattering theory,
and the corrections come from the contribution of newly introduced operators living at threshold energies.
For the model presented here, these corrections appear both at $0$-energy
and at energy corresponding to $+\infty$. Accordingly, the number of eigenvalues of $H$ will be equal to the winding
number of the scattering operator together with the contributions of an operator living at $0$-energy and of an operator living at $\infty$-energy.
This relation and its explanation correspond to our main result.

This note is organized as follows: In Section \ref{sec_model} we introduce the model and recall
some relations obtained in \cite{DR}. The main result is presented in Section \ref{sec_results} which
also contains some additional information on the wave operators.
Finally, the algebraic framework and the few analytical proofs are provided in Section~\ref{sec_proof}.

\section{The model}\label{sec_model}

The material of this section is borrowed from the reference \cite{DR}.
For any $m \in \C$ we consider the differential expression
\begin{equation*}
L_{m^2} :=-\partial_x^2+\big(m^2-\frac14\big)\frac{1}{x^{2}}\ .
\end{equation*}
The minimal and maximal operators associated with it in $L^2(\R_+)$
are given by
$\Dom(L_{m^2}^{\max}) = \{f\in L^2(\R_+)\mid L_{m^2}  f\in L^2(\R_+)\}$
and $\Dom(L_{m^2}^{\min})$ is
the closure of the restriction of $L_{m^2}$ to $C_{\rm c}^\infty(\R_+)$.
These realizations satisfy the relation $\big(L_{m^2}^{\min}\big)^* = L_{\bar m^2}^{\max}$.
In addition if $|\Re(m)|<1$ then
$L_{m^2}^{\min} \subsetneq L_{m^2}^{\max}$ and $\Dom(L_{m^2}^{\min})$ is a closed subspace of codimension $2$ of $\Dom(L_{m^2}^{\max})$.
More precisely, if $|\Re(m)|\in (0,1)$ and if $f\in \Dom(L_{m^2}^{\max})$, then there exist $a,b\in \C$
such that:
\begin{equation*}
f(x)  - ax^{1/2-m} - bx^{1/2+m} \in\Dom(L_{m^2}^{\min})\hbox{
around }0.
\end{equation*}
Here the expression \emph{$g(x)\in \Dom(L_{m^2} ^{\min})$ around $0$} means that there exists
$\zeta\in C_{\rm c}^\infty\big([0,\infty)\big)$ with $\zeta=1$ around $0$ such that $g\zeta\in
\Dom(L_{m^2} ^{\min})$.
Thus, for any $\kappa\in \C$ we define the family of operators $H_{m,\kappa}$~:
\begin{align*}
\Dom(H_{m,\kappa}) & = \big\{f\in \Dom(L_{m^2}^{\max})\mid
\hbox{ for some }  c \in \C,\\
&\qquad f(x)- c
\big(\kappa x^{1/2-m} +
x^{1/2+m}\big)\in\Dom(L_{m^2}^{\min})\hbox{ around }
0\big\}.
\end{align*}
Note that for simplicity the special case $\Re(m)=0$ is not considered in the present manuscript.

Many properties of these operators have been exhibited in \cite{DR}.
For the spectral theory let us simply mention that $H_{m,\kappa}$ is self-adjoint
if and only if $m$ and $\kappa$ are real. The operators $H_{m,\kappa}$
have a finite number of eigenvalues which are located in $\C\setminus [0,\infty)$, and
in addition one has $[0,\infty)\subset \sigma(H_{m,\kappa})$.
A limiting absorption principle has been shown for these operators on $(0,\infty)$, with a slight restriction
if $(m,\kappa)$ is an exceptional pair.
We say that $(m,\kappa)$ is \emph{an exceptional pair}
if $\kappa\neq 0$ and $\pm \pi \in \Im\big(\frac{1}{m}\Ln(\varsigma)\big)$ with
\begin{equation*}
\varsigma:=\kappa \frac{\Gamma(-m)}{\Gamma(m)},
\end{equation*}
and where $\Gamma$ is the usual $\Gamma$-function and $\Ln$ the multivalued logarithm.

Still in the non-exceptional case, an \emph{incoming} and
an \emph{outgoing Hankel transformations} $\cF_{m,\kappa}^\mp$ can be defined on $C^\infty_{\rm c}(\R_+)$
by the kernels
\begin{equation*}
\cF_{m,\kappa}^\mp(x,y) :=\e^{\mp \i \frac{\pi}{2}m}\sqrt{\frac2\pi}
\frac{\Ja_{m}(xy)- \varsigma \Ja_{-m}(xy)\big(\tfrac{y^2}{4}\big)^{m}}{
1-\varsigma \e^{\mp\i \pi m}\big(\tfrac{y^2}{4}\big)^{m}}\ .
\end{equation*}
with $\Ja_m(z):= \sqrt{\frac{\pi z}{2}} J_m(z)$ and $J_m$ the usual
Bessel function.
These operators extend then continuously to elements of $\B\big(L^2(\R_+)\big)$.
Now, for any bounded operator $B$ with integral kernel $B(x,y)$ let us set
$B^\t$ for its transpose, \ie~for the operator satisfying $B^\t(x,y)=B(y,x)$.
With this notation the operators $\cF_{m,\kappa}^\pm$ satisfy
$\cF_{m,\kappa}^{\pm \t} \cF_{m,\kappa}^{\mp}=\one$.
In addition if we set
\begin{equation*}
\one_{\R_+}(H_{m,\kappa}) := \cF_{m,\kappa}^{+}\;\!\cF_{m,\kappa}^{- \t}
=  \cF_{m,\kappa}^{-}\;\!\cF_{m,\kappa}^{+ \t}
\end{equation*}
one observes that this operator is a projection. The following equalities have also been proved
in \cite{DR}:
For any $k\in \C$ with $\Re(k)>0$ and $-k^2\not \in \sigma_{\rm p}(H_{m,\kappa})$ one has
\begin{align*}
(H_{m,\kappa}+k^2)^{-1} \one_{\R_+}(H_{m,\kappa})
& =\cF_{m,\kappa}^{\pm}(\Q^2+k^2)^{-1} \cF_{m,\kappa}^{\mp \t} \\
& =\one_{\R_+}(H_{m,\kappa})(H_{m,\kappa}+k^2)^{-1}
\end{align*}
where $\Q$ is the usual operator of multiplication by the variable in $L^2(\R_+)$.

Finally, whenever $(m,\kappa)$ and $(m',\kappa')$ are not exceptional pairs,
the wave operators for the pair of operators $(H_{m,\kappa},H_{m',\kappa'})$ can be defined by
\begin{equation*}
W_{m,\kappa;m',\kappa'}^\pm :=  \cF_{m,\kappa}^{\pm}\;\! \cF_{m',\kappa'}^{\mp \t}\;\! .
\end{equation*}
These operators satisfy the relations
\begin{equation}\label{eq_rel1}
W_{m,\kappa;m',\kappa'}^{\mp \t} \;\!W_{m,\kappa;m',\kappa'}^\pm  = \one_{\R_+}(H_{m',\kappa'})
\end{equation}
and
\begin{equation}\label{eq_rel2}
W_{m,\kappa;m',\kappa'}^\pm \;\!W_{m,\kappa;m',\kappa'}^{\mp \t} = \one_{\R_+}(H_{m,\kappa})
\end{equation}
as well as the intertwining relation
\begin{equation*}
W_{m,\kappa;m',\kappa'}^\pm H_{m',\kappa'} = H_{m,\kappa}W_{m,\kappa;m',\kappa'}^\pm\;\! .
\end{equation*}
The scattering operator is finally defined by
\begin{equation*}
S_{m,\kappa;m',\kappa'}:=W_{m,\kappa;m',\kappa'}^{-\t}W_{m,\kappa;m',\kappa'}^-\;\!.
\end{equation*}

\section{The main result}\label{sec_results}

From now on, let us fix a pair $(m,\kappa)$ with $|\Re(m)|\in (0,1)$ and $\kappa\in \C$
which is not exceptional.
For the reference system we consider one of the simplest one, namely the Dirichlet
Laplacian $H_{\rm D}$ on the half-line. This operator is obtained from the general family
for the indices $(\frac{1}{2},0)$. Note that by the chain-rule any other
pair $(m',\kappa')$ can easily be used for the reference system.

In the following statement, we  provide an alternative representation of the wave operator
which is at the root of the algebraic framework presented in \cite{R16}.
For that purpose, note first that for the special choice $(\frac{1}{2},0)$ the transformation $\cF^{+\t}_{\frac{1}{2},0}\equiv \cF^+_{\frac{1}{2},0}$
corresponds to $\e^{\i\frac{\pi}{4}}\cF_{\rm D}$ with
$\cF_{\rm D}$ the usual Fourier sine transformation on $\R_+$.
Let us also define for any $t\in \R$ the function $\Xi_m$ given by
\begin{equation*}
\Xi_m(t):= \e^{\i\ln(2)t} \frac{\Gamma(\frac{m+1+\i t}{2})}{\Gamma(\frac{m+1-\i t}{2})}\ .
\end{equation*}
We finally consider the unitary group $\{U_t\}_{t\in \R}$ acting on any $f\in L^2(\R_+)$ as
\begin{equation*}
[U_tf](x) = \e^{t/2} f\big(\e^t x\big), \qquad \forall x\in \R_+
\end{equation*}
which is usually called \emph{the unitary group of dilations}.
Its self-adjoint generator is denoted by $A$ and is called \emph{the generator of dilations}.

\begin{lemma}\label{lem_wave}
If $m\in \C$ with $|\Re(m)|\in (0,1)$ and if $(m,\kappa)$ is not an exceptional pair then the operator
$W^-_{m,\kappa;\frac{1}{2},0}$ is equal to
\begin{equation}\label{eq_wave_op}
\e^{\i\frac{\pi}{4}}  \Xi_{\frac{1}{2}}(A) \Big(\Xi_m(-A)-\varsigma\Xi_{-m}(-A) \big(\tfrac{H_{\rm D}}{4}\big)^{m}\Big)
\frac{\e^{- \i \frac{\pi}{2}m}}{1-\varsigma \e^{- \i \pi m} \big(\tfrac{H_{\rm D}}{4}\big)^{m}}\ .
\end{equation}
\end{lemma}

Motivated by the formula obtained in the previous statement, let us now define the function of two variables:
$\Gamma_{m,\kappa;\frac{1}{2},0} : \R_+ \times \R$ by
\begin{equation*}
\Gamma_{m,\kappa;\frac{1}{2},0}(x,t) :=  \e^{\i\frac{\pi}{4}}  \Xi_{\frac{1}{2}}(t) \Big(\Xi_m(-t)-\varsigma\Xi_{-m}(-t) \big(\tfrac{x^2}{4}\big)^{m}\Big)
\frac{\e^{- \i \frac{\pi}{2}m}}{1-\varsigma \e^{- \i \pi m} \big(\tfrac{x^2}{4}\big)^{m}}\ .
\end{equation*}
Note that the condition $(m,\kappa)$ is not an exceptional pair precisely prevents the denominator in the last factor
to vanish.
In addition, it is easily observed that this function is continuous on the square $\blacksquare:=[0,+\infty]\times [-\infty,+\infty]$,
and therefore its restriction on the boundary $\square$ of the square is also well defined and continuous.
Note that this boundary is made of four parts: $\square = B_1\cup B_2 \cup B_3 \cup B_4$
with $B_1 = \{0\}\times[-\infty,+\infty]$, $B_2 =[0,+\infty]\times\{+\infty\}$, $B_3 = \{+\infty\}\times[-\infty,+\infty]$,
and $B_4=[0,+\infty]\times\{-\infty\}$. Thus, the algebra $C(\square)$ of continuous functions on $\square$
can be viewed as a subalgebra of
\begin{equation}\label{eq_alg_sum}
C\big([-\infty,+\infty]\big)\oplus C\big([0,+\infty]\big)\oplus C\big([-\infty,+\infty]\big)\oplus C\big([0,+\infty]\big)
\end{equation}
given by elements $(\Gamma_1, \Gamma_2, \Gamma_3, \Gamma_4)$ which coincide at the corresponding end points, that
is,
$\Gamma_1(+\infty) = \Gamma_2(0)$, $\Gamma_2(+\infty) = \Gamma_3(+\infty)$, $\Gamma_3(-\infty) = \Gamma_4(+\infty)$, and
$\Gamma_4(0) = \Gamma_1(-\infty)$.
As shown in the following section one gets for $\kappa\neq 0$
\begin{align}
\label{eq_1} \Gamma_1(t) & :=\Gamma_{m,\kappa;\frac{1}{2},0}(0,t) = \left\{\begin{matrix}  \e^{\i\frac{\pi}{2}(\frac{1}{2}-m)}  \Xi_{\frac{1}{2}}(t) \Xi_m(-t)
& \hbox{ if } \ \Re(m)>0, \\
\e^{\i\frac{\pi}{2}(\frac{1}{2}+m)}  \Xi_{\frac{1}{2}}(t) \Xi_{-m}(-t) & \hbox{ if } \ \Re(m)<0,
\end{matrix}\right.\\
\label{eq_scat_op}  \Gamma_2(x) & := \Gamma_{m,\kappa;\frac{1}{2},0}(x,+\infty) =
e^{\i\pi(\frac{1}{2}-m)} \frac{1 -\varsigma \e^{+\i \pi m} \big(\tfrac{x^2}{4}\big)^{m} }{1-\varsigma \e^{- \i \pi m} \big(\tfrac{x^2}{4}\big)^{m}}, \\
\Gamma_3(t) & :=\Gamma_{m,\kappa;\frac{1}{2},0}(+\infty,t) = \left\{\begin{matrix} \e^{\i\frac{\pi}{2}(\frac{1}{2}+m)}  \Xi_{\frac{1}{2}}(t) \Xi_{-m}(-t) & \hbox{ if } \ \Re(m)>0, \\
\e^{\i\frac{\pi}{2}(\frac{1}{2}-m)}  \Xi_{\frac{1}{2}}(t) \Xi_m(-t)  & \hbox{ if } \ \Re(m)<0,
\end{matrix}\right.\\
\label{eq_2} \Gamma_4(x) & := \Gamma_{m,\kappa;\frac{1}{2},0}(x,-\infty) = 1.
\end{align}
In the special case $\kappa=0$ one has $\Gamma_1(t)=\Gamma_3(t)=\e^{\i\frac{\pi}{2}(\frac{1}{2}-m)}  \Xi_{\frac{1}{2}}(t) \Xi_m(-t)$,
$\Gamma_2(x)= \e^{\i\pi(\frac{1}{2}-m)}$ and $\Gamma_4(x)=1$.

Let us now observe that the boundary $\square$  of $\blacksquare$ is homeomorphic to the circle $\S$.
Observe in addition that the restriction $\Gammaq_{m,\kappa;\frac{1}{2},0}$ of the function $\Gamma_{m,\kappa;\frac{1}{2},0}$ to $\square$ takes its values in $\C\setminus \{0\}$.
Then, since $\Gammaq_{m,\kappa;\frac{1}{2},0}$ is a continuous function on the closed curve $\square$ and takes
non-zero values, its winding number $\Wind(\Gammaq_{m,\kappa;\frac{1}{2},0})$ is well defined.
By convention, we shall turn around $\square$ clockwise
and the increase in the winding number is also counted clockwise. Let us stress that the contribution on $B_3$
has to be computed from $+\infty$ to $-\infty$, and the contribution on $B_4$ from $+\infty$ to $0$.
Our main result now reads:

\begin{theorem}\label{thm_main}
If $m\in \C$ with $|\Re(m)|\in (0,1)$ and if $(m,\kappa)$ is not an exceptional pair then
\begin{equation}\label{eq_Lev_1}
\Wind(\Gammaq_{m,\kappa;\frac{1}{2},0}) = \hbox{ number of eigenvalues of } H_{m,\kappa}\;\!.
\end{equation}
\end{theorem}

The above statement is a topological version of Levinson's theorem, and corresponds to an index theorem.
In order to make the link with the usual formulation, it is necessary to consider the 4 contributions $\Gamma_j$
separately. Since $\Gamma_1(t) = \Gamma_{m,\kappa;\frac{1}{2},0}(0,t)$, the corresponding operator $\Gamma_1(A)$ can be understood
as an operator related to the $0$-energy of the scattering process. Similarly, $\Gamma_3(A)$ is an operator
associated with the energy $\infty$ of the scattering process.
On the other hand, the expression in \eqref{eq_scat_op} corresponds to the scattering operator, or more precisely one has
$\Gamma_2(\sqrt{H_{\rm D}})= S_{m,\kappa;\frac{1}{2},0} \equiv  S_{m,\kappa;\frac{1}{2},0}(H_{\rm D})$, where
$S_{m,\kappa;\frac{1}{2},0}$ is the scattering operator for the pair $(H_{m,\kappa},H_{\rm D})$,
see \cite[Sec.~6.5]{DR}.
Now, by a slight adaptation of the proof of \cite[Lem.~4]{KPR} the contribution to the winding number coming from $\Gamma_1$ and $\Gamma_3$
are respectively equal to $\frac{\Re(m)}{2}-\frac{1}{4}$ and $\frac{\Re(m)}{2}+\frac{1}{4}$ when $\Re(m)>0$.
Similarly, for $\Re(m)<0$ these two contributions are respectively equal to $-\frac{\Re(m)}{2}-\frac{1}{4}$ and $-\frac{\Re(m)}{2}+\frac{1}{4}$.
By collecting these information one gets:

\begin{corollary}\label{corol_Lev}
If $m\in \C$ with $|\Re(m)|\in (0,1)$, if $\kappa \neq 0$ and if $(m,\kappa)$ is not an exceptional pair one has
\begin{equation}\label{eq_precise}
\Wind\big(S_{m,\kappa;\frac{1}{2},0}(\cdot)\big) +|\Re(m)| = \hbox{ number of eigenvalues of } H_{m,\kappa}\;\!.
\end{equation}
\end{corollary}

Note that the special case $\kappa=0$ is less interesting since $S_{m,0;\frac{1}{2},0}$ is a constant
and both sides of \eqref{eq_Lev_1} in this case are equal to $0$.

\begin{remark}
For general $m$ and $\kappa$ in the range mentioned above, the computation of the r.h.s.~of \eqref{eq_precise} is rather involved
and can take arbitrary large but finite values.
Some useful information are provided in \cite[Prop.~5.3 \& Lem.~5,4]{DR}.
Equivalently, the computation of the winding number of the map $x\mapsto S_{m,\kappa;\frac{1}{2},0}(x)$ is fairly tricky
for complex numbers $m$ and $\kappa$. Anyway, we can directly see from Corollary \ref{corol_Lev} that the computation of this winding
number does not provide enough information for deducing the number of bound states of $H_{m,\kappa}$,
and that the correction $+|\Re(m)|$ can take arbitrary values in $(0,1)$.
\end{remark}

\section{The proofs}\label{sec_proof}

Before introducing the necessary algebraic framework, let us check the analytical part of our investigations.
More precisely let us check some properties of the wave operators and of the function $\Gamma_{m,\kappa;\frac{1}{2},0}$.

\begin{proof}[Proof of Lemma \ref{lem_wave}]
We first look at the wave operator in the spectral representation of $H_{\rm D}$,
or more precisely let us consider the operator
\begin{equation*}
\cF^{*}_{\rm D}W^-_{m,\kappa;\frac{1}{2},0} \cF_{\rm D}
= \cF^{*}_{\rm D} \cF^-_{m,\kappa} \cF^{+\t}_{\frac{1}{2},0}  \cF_{\rm D}
= \e^{\i\frac{\pi}{4}} \cF_{\rm D}^* \cF^-_{m,\kappa}
\end{equation*}
with
\begin{equation}\label{eq_new_1}
\cF^*_{\rm D} = \cF_{\rm D}=  \Xi_{\frac{1}{2}}(-A) \Ju
\end{equation}
and
\begin{equation}\label{eq_new_2}
\cF^-_{m,\kappa} =  \Ju \Big(\Xi_m(A)-\varsigma\Xi_{-m}(A) \big(\tfrac{\Q^2}{4}\big)^{m}\Big)
\frac{\e^{- \i \frac{\pi}{2}m}}{1-\varsigma \e^{- \i \pi m} \big(\tfrac{\Q^2}{4}\big)^{m}}\ .
\end{equation}
The unitary and self-adjoint transformation
$\Ju:L^2(\R_+)\to L^2(\R_+)$  is defined by the formula
$\big(\Ju f\big)(x)=\frac{1}{x}f\big(\frac{1}{x}\big)$
for any $f\in L^2(\R_+)$ and $x\in \R_+$
Note that equalities \eqref{eq_new_1} and \eqref{eq_new_2} have been obtained in \cite[Prop.~4.5 \& Lem.~6.3]{DR}.
By taking into account the relations $\cF^{*}_{\rm D}A\cF_{\rm D}=-A$
and $\cF^{*}_{\rm D}\Q^2\cF_{\rm D}=H_{\rm D}$ one directly deduces formula \eqref{eq_wave_op}.
\end{proof}

For showing that the function $\Gamma_{m,\kappa;\frac{1}{2},0}$ extends to the square $[0,+\infty]\times [-\infty,+\infty]$
let us compute its asymptotics.
For the limits we recall from \cite[Eq.~4.25]{DR} that
\begin{equation*}
\Xi_\frac{1}{2}(\mp\infty)\Xi_{m'}(\pm\infty)
=\e^{\mp\i \frac{\pi}{2}(\frac{1}{2}-m')}.
\end{equation*}
One thus infers that for $\kappa \neq 0$
\begin{align*}
\Gamma_{m,\kappa;\frac{1}{2},0}(x,-\infty) & = \e^{\i\frac{\pi}{4}} \Big(\e^{-\i \frac{\pi}{2}(\frac{1}{2}-m)}
-\varsigma \e^{-\i \frac{\pi}{2}(\frac{1}{2}+m)} \big(\tfrac{x^2}{4}\big)^{m}\Big)
\frac{\e^{- \i \frac{\pi}{2}m}}{1-\varsigma \e^{- \i \pi m} \big(\tfrac{x^2}{4}\big)^{m}} \\
& =\frac{1 -\varsigma \e^{- \i \pi m} \big(\tfrac{x^2}{4}\big)^{m} }{1-\varsigma \e^{- \i \pi m} \big(\tfrac{x^2}{4}\big)^{m}} \\
& = 1
\end{align*}
while
\begin{align*}
\Gamma_{m,\kappa;\frac{1}{2},0}(x,+\infty) & = \e^{\i\frac{\pi}{4}} \Big(\e^{+\i \frac{\pi}{2}(\frac{1}{2}-m)}
-\varsigma \e^{+\i \frac{\pi}{2}(\frac{1}{2}+m)} \big(\tfrac{x^2}{4}\big)^{m}\Big)
\frac{\e^{- \i \frac{\pi}{2}m}}{1-\varsigma \e^{- \i \pi m} \big(\tfrac{x^2}{4}\big)^{m}} \\
& =
e^{\i\pi(\frac{1}{2}-m)} \frac{1 -\varsigma \e^{+\i \pi m} \big(\tfrac{x^2}{4}\big)^{m} }{1-\varsigma \e^{- \i \pi m} \big(\tfrac{x^2}{4}\big)^{m}}\ .
\end{align*}

For the other two limits one gets if $\Re(m)>0$
\begin{equation*}
\Gamma_{m,\kappa;\frac{1}{2},0}(0,t) =\lim_{x\to 0}\Gamma_{m,\kappa;\frac{1}{2},0}(x,t)= \e^{\i\frac{\pi}{2}(\frac{1}{2}-m)}  \Xi_{\frac{1}{2}}(t) \Xi_m(-t),
\end{equation*}
while if $\Re(m)<0$ one has $\Gamma_{m,\kappa;\frac{1}{2},0}(0,t) = \e^{\i\frac{\pi}{2}(\frac{1}{2}+m)}  \Xi_{\frac{1}{2}}(t) \Xi_{-m}(-t)$.
On the other hand, if $\Re(m)>0$ one gets
\begin{equation*}
\Gamma_{m,\kappa;\frac{1}{2},0}(+\infty,t) = \lim_{x\to +\infty}\Gamma_{m,\kappa;\frac{1}{2},0}(x,t)
= \e^{\i\frac{\pi}{2}(\frac{1}{2}+m)}  \Xi_{\frac{1}{2}}(t) \Xi_{-m}(-t)
\end{equation*}
while if $\Re(m)<0$ one has
$\Gamma_{m,\kappa;\frac{1}{2},0}(+\infty,t) = \e^{\i\frac{\pi}{2}(\frac{1}{2}-m)}  \Xi_{\frac{1}{2}}(t) \Xi_m(-t)$.

The above expressions are summarized in equations \eqref{eq_1} to \eqref{eq_2}.
The case $\kappa=0$ can be obtained similarly, and is simpler.

Let us now discuss the properties of the kernel and of the range of the wave operator $W^-_{m,\kappa;\frac{1}{2},0}$.
By taking into account \cite[Thm.~III.6.17]{Kato} on the separation of the spectrum into two components,
and by using an alternative definition of the projection $\one_{\R_+}(H_{m,\kappa})$ in terms of the spectral density integrated on $\R_+$, as
provided in \cite[Sec.~6.4]{DR}, one infers that the subspace defined by $\one_{\R_+}(H_{m,\kappa})$ is complementary to the
subspace generated by the eigenfunctions of the operator $H_{m,\kappa}$. Since these eigenvalues are always in finite number
the codimension of $\one_{\R_+}(H_{m,\kappa})L^2(\R_+)$ is always finite.
Similarly, since $H_{\rm D}$ has no eigenvalue one gets $\one_{\R_+}(H_{\rm D})=\one$.
Then, by the properties \eqref{eq_rel1} and \eqref{eq_rel2} one infers that $W^-_{m,\kappa;\frac{1}{2},0}$ is a Fredholm
operator.
Summing up these information one gets~:

\begin{lemma}\label{lem_Fred}
Let $m\in \C$ with $|\Re(m)|\in (0,1)$ and assume that $(m,\kappa)$ is not an exceptional pair.
Then $W^-_{m,\kappa;\frac{1}{2},0}$ is a Fredholm operator with
\begin{equation*}
\dim\ker\big(W^-_{m,\kappa;\frac{1}{2},0}\big)
-\dim\coker\big(W^-_{m,\kappa;\frac{1}{2},0}\big) = -\# \sigma_{{\rm p}}(H_{m,\kappa}).
\end{equation*}
\end{lemma}

The content of Theorem \ref{thm_main} can now be deduced either from \cite[Thm.~4.4]{R16} (with $n=-1$) or from
\cite[Thm.~3]{CH}. Let us however make some comments. In \cite{R16} statements similar to Theorem \ref{thm_main} are provided for various models
related to quantum mechanics. The main difference with these models is that here the operator $H_{m,\kappa}$
is not always self-adjoint. It follows that the wave operator $W^-_{m,\kappa;\frac{1}{2},0}$ is not always an isometry but only a Fredholm operator,
and accordingly the operator $\Gammaq_{m,\kappa;\frac{1}{2},0}$ is not always a unitary operator but is still an invertible operator.
However, for the algebraic construction used for obtaining a topological version of Levinson's theorem,
these differences are perfectly manageable and the construction works in this case as well, see for example \cite[Rem.~8.1.7]{RLL}.

The main ingredient for the algebraic construction consists first in exhibiting a natural subalgebra of $\B\big(L^2(\R_+)\big)$
which contains the wave operator $W^-_{m,\kappa;\frac{1}{2},0}$.
By looking at the special representation obtained for the wave operator in \eqref{eq_wave_op}
one deduces that $W^-_{m,\kappa;\frac{1}{2},0}$ belongs to the $C^*$-algebra $\EE_{(H_{\rm D},A)}$
introduced in \cite[Sec.~4.4]{R16}. This $C^*$-algebra is generated by product of the form
$\psi(H_{\rm D})\eta(A)$ with $\psi \in C\big([0,\infty]\big)$ and $\eta\in C\big([-\infty,\infty]\big)$.
Our interest in this algebra comes from its easily understandable quotient through the ideal $\K\big(L^2(\R_+)\big)$ of compact operators
on $L^2(\R_+)$. In fact, this quotient is isomorphic to the $C^*$-algebra $C(\square)$ which can be viewed as a subalgebra
of the one introduced in \eqref{eq_alg_sum}. Thus, if we set $q: \EE_{(H_{\rm D},A)} \to \EE_{(H_{\rm D},A)}/\K\big(L^2(\R_+)\big)$ for
the quotient map, one gets that $q\big(W^-_{m,\kappa;\frac{1}{2},0}\big) = \Gammaq_{m,\kappa;\frac{1}{2},0}= (\Gamma_1,\Gamma_2,\Gamma_3,\Gamma_4)$ with
$\Gamma_j$ introduced in \eqref{eq_1} to \eqref{eq_2}.
The final argument consists in borrowing the information from \cite[Thm.~4.4]{R16} that the winding number of $\Gammaq_{m,\kappa;\frac{1}{2},0}$
is equal to (minus) the Fredholm index of the operator $W^-_{m,\kappa;\frac{1}{2},0}$. By using Lemma \ref{lem_Fred} this leads directly to Theorem \ref{thm_main}.

Alternatively, one can directly use the content of Theorem 3 of \cite{CH} once it is observed that the algebra presented
in that paper corresponds to the algebra $\EE_{(\Q,A)}$ introduced in \cite[Sec.~4.4]{R16}. This algebra is isomorphic to the $C^*$-algebra
$\EE_{(H_{\rm D},A)}$ by a conjugation with the unitary map $\cF_{\rm D}$. The difference of a minus sign between the content of \cite[Thm.~3]{CH}
and Theorem \ref{thm_main} comes from the equality $\cF_{\rm D}^*A\cF_{\rm D} = -A$ which reverses part of the construction.
Finally, let us also mention Theorem 4 in \cite{BC} which provides a similar abstract result but in a larger setting.

\begin{remark}
A more analytical approach involving Jost function could also be used
for proving Corollary \ref{corol_Lev}. However, the flavor of an index theorem
would be lost, and the contributions of the operators at $0$-energy
and at energy $+\infty$ would not appear so explicitly.
\end{remark}

\end{document}